%% file: main.tex
\documentclass{article}
\usepackage{amsthm,amssymb}
\usepackage{MnSymbol}
\usepackage[margin=1in]{geometry}

\usepackage{graphicx}
\graphicspath{{figures/}}
\usepackage{xcolor}
\usepackage[ruled,linesnumbered,noend]{algorithm2e}
\usepackage{eqparbox}
\usepackage{float}
\usepackage{parskip}
\usepackage{mathtools}
\usepackage{enumitem}
\usepackage{booktabs}
\usepackage{subcaption}
\definecolor{linkblue}{HTML}{4A86B8}
\usepackage[colorlinks,linkcolor=linkblue,citecolor=linkblue,urlcolor=linkblue]{hyperref}
\usepackage[longnamesfirst]{natbib}
\usepackage{complexity}
\usepackage{tikz}
\usetikzlibrary{arrows.meta}

\input{commands}

\usepackage{aliascnt}
\usepackage[capitalize,nameinlink]{cleveref}
\usepackage{autonum}

\theoremstyle{plain}
\newtheorem{maintheorem}{Main Theorem}
\crefname{maintheorem}{main theorem}{main theorems}
\Crefname{maintheorem}{Main Theorem}{Main Theorems}
\newtheorem{theorem}{Theorem}[section]
\newaliascnt{lemma}{theorem}
\newtheorem{lemma}[lemma]{Lemma}
\aliascntresetthe{lemma}
\crefname{lemma}{lemma}{lemmas}
\Crefname{lemma}{Lemma}{Lemmas}
\newaliascnt{proposition}{theorem}
\newtheorem{proposition}[proposition]{Proposition}
\aliascntresetthe{proposition}
\crefname{proposition}{proposition}{propositions}
\Crefname{proposition}{Proposition}{Propositions}
\newaliascnt{corollary}{theorem}

\aliascntresetthe{corollary}
\crefname{corollary}{corollary}{corollaries}
\Crefname{corollary}{Corollary}{Corollaries}
\newaliascnt{claim}{theorem}
\newtheorem{claim}[claim]{Claim}
\aliascntresetthe{claim}
\crefname{claim}{claim}{claims}
\Crefname{claim}{Claim}{Claims}
\theoremstyle{definition}
\newaliascnt{definition}{theorem}
\newtheorem{definition}[definition]{Definition}
\aliascntresetthe{definition}
\crefname{definition}{definition}{definitions}
\Crefname{definition}{Definition}{Definitions}
\newaliascnt{assumption}{theorem}
\newtheorem{assumption}[assumption]{Assumption}
\aliascntresetthe{assumption}
\crefname{assumption}{assumption}{assumptions}
\Crefname{assumption}{Assumption}{Assumptions}
\newaliascnt{remark}{theorem}
\newtheorem{remark}[remark]{Remark}
\aliascntresetthe{remark}
\crefname{remark}{remark}{remarks}
\Crefname{remark}{Remark}{Remarks}
\newaliascnt{example}{theorem}

\aliascntresetthe{example}
\crefname{example}{example}{examples}
\Crefname{example}{Example}{Examples}

\usepackage{thm-restate}

\title{Fixed Sets and the Hardness of Computing Correlated Equilibria\thanks{This paper first appeared online on October 20, 2026.}}
\author{%
    Ioannis Anagnostides\thanks{\protect\label{aff:cmu}Carnegie Mellon University. \texttt{ianagnos@cs.cmu.edu}, \texttt{sandholm@cs.cmu.edu}.}
    \and Constantinos Daskalakis\thanks{\protect\label{aff:mit}Massachusetts Institute of Technology. \texttt{costis@mit.edu}, \texttt{gfarina@mit.edu}, \texttt{zhangbh@mit.edu}.}
    \and Gabriele Farina\textsuperscript{\ref*{aff:mit}}
    \and Noah Golowich\thanks{University of Texas, Austin. \texttt{nzg@cs.utexas.edu}.}
    \and Tuomas Sandholm\textsuperscript{\ref*{aff:cmu},}\thanks{Additional affiliations: Strategy Robot, Inc., Strategic Machine, Inc., Optimized Markets, Inc.}
    \and Brian Hu Zhang\textsuperscript{\ref*{aff:mit}}
}
\date{}

\begin{document}
\maketitle
\begin{abstract}
    We show that computing an approximate correlated equilibrium of a polytope game is \textsf{PPAD}-complete, and thus polynomial-time equivalent to finding a Nash equilibrium. This challenges the conventional view in game theory that correlated equilibria are more tractable objects than Nash equilibria.
    
    From a conceptual standpoint, we connect correlated equilibria to a natural generalization of fixed points, which we call \emph{fixed sets}. A fixed set of a continuous function $f : \mathcal{X} \to \mathcal{X}$ is a nonempty finite set $S \subseteq \mathcal{X}$ such that $S \subseteq \operatorname{conv} f(S)$. Our main technical contribution is to prove that computing a fixed set is \PPAD-hard even at inverse polynomial precision.
\end{abstract}

\input{text/introduction.tex}

\input{text/preliminaries.tex}

\input{text/fixed-sets.tex}

\input{text/amplification.tex}

\input{text/multilinear-games.tex}

\input{text/ai-disclosure.tex}

\section*{Acknowledgments}

I.A. and T.S. are supported by NIH award A240108S001, the Vannevar Bush Faculty Fellowship ONR N00014-23-1-2876, and National Science Foundation grant RI-2312342. I.A. thanks Aviad Rubinstein for a discussion concerning fixed sets. C.D. was supported by a Simons Investigator Award, a Simons Collaboration on Algorithmic Fairness, ONR MURI grant N00014-25-1-2116, ONR grant N00014-25-1-2296. G.F. was supported in part by the National Science Foundation award CCF-2443068, the Office of Naval Research grant N000142512296, and an AI2050 Early Career Fellowship.

\bibliographystyle{plainnat}
\bibliography{dairefs}

\end{document}

%% file: commands.tex
\usepackage{bm}
\usepackage{bbm}
\usepackage{xspace}
\usepackage{complexity}
\usepackage{amsthm}

\usepackage{physics2}
\usephysicsmodule{qtext.legacy,ab}
\let\mid\mvert
\usepackage{mleftright}
\mleftright

\newcommand{\delimit}[3]{
    \NewDocumentCommand#1{sO{}m}{
    \IfBooleanTF##1
    {\left#2 ##3 \right#3}
    {\mathopen{##2 #2} ##3 \mathclose{##2 #3}}%
    }
}

\delimit\ceil\lceil\rceil
\delimit\floor\lfloor\rfloor
\delimit\ip\langle\rangle
\delimit\norm\lVert\rVert
\delimit\abs\lvert\rvert

\def\va{{\bm{a}}}

\def\ve{{\bm{e}}}

\def\vs{{\bm{s}}}
\def\vt{{\bm{t}}}

\def\vv{{\bm{v}}}
\def\vw{{\bm{w}}}
\def\vx{{\bm{x}}}
\def\vy{{\bm{y}}}
\def\vz{{\bm{z}}}

\renewcommand\vec\bm

\def\mI{{\mathbf{I}}}

\def\mK{{\mathbf{K}}}

\def\mW{{\mathbf{W}}}

\def\cS{{\mathcal{S}}}

\def\cX{{\mathcal{X}}}

\let\E\relax
\DeclareMathOperator*{\E}{\mathbb E}

\DeclareMathOperator*{\clip}{clip}

\renewcommand{\R}{\mathbb R}

\newcommand{\Z}{\mathbb Z}

\let\op\operatorname
\let\eps\epsilon

\newcommand{\supp}{\op{supp}}

\newcommand{\defeq}{\coloneqq}
\newcommand{\eqdef}{\eqqcolon}

\newcommand{\eg}{{\em e.g.}\xspace}

\newcommand{\Geo}{\op{Geo}}

\newcommand{\KL}{d_\text{KL}}

%% file: text/introduction.tex
\section{Introduction}

The \emph{correlated equilibrium (CE)}, introduced by~\citet{Aumann74:Subjectivity}, is a central solution concept in economics. It is a generalization of \emph{Nash equilibrium}~\citep{Nash50:Equilibrium}, enabling correlation between the players' strategies, and is in many ways more natural than its predecessor: in the words of Roger Myerson, ``if there is intelligent life on other planets, in a majority of them, they would have discovered correlated equilibrium before Nash equilibrium''~\citep[p.~83]{Shoham09:Multiagent}. Any student of game theory learns early that a correlated equilibrium can be obtained as a solution to a linear program. For an explicitly represented game in normal form, this yields a polynomial-time algorithm. Taking a step further, a celebrated result by~\citet{Papadimitriou08:Computing} extended this tractability frontier even in multi-player \emph{succinct} games, in which case the linear program describing the set of correlated equilibria contains exponentially many variables. Their algorithm---aptly named ``ellipsoid against hope''---assumes a natural condition they refer to as the \emph{polynomial expectation property}, which requires that expected utilities under product distributions can be computed in polynomial time. Another crucial premise in their algorithm is that the game is of \emph{polynomial type}, meaning that the number of pure strategies available to each player is polynomial in the description. However, this is rarely the case in realistic strategic interactions, where players have complex strategy spaces. The complexity of computing correlated equilibria beyond games of polynomial type has remained an enigmatic question, despite the fact that correlated equilibria have played a central role in the development of game theory for more than 50 years.

In stark contrast to correlated equilibria, which can be phrased as solutions to linear programs and thereby form a convex set, Nash equilibria are topologically more complex objects~\citep{Datta03:Universality,Balkenborg19:On}. Indeed, Nash's famous proof of existence hinges on a \emph{fixed-point} argument, and in fact, fixed points are computationally equivalent to Nash equilibria~\citep{Daskalakis08:Complexity,Chen09:Settling,Etessami07:Complexity}. These celebrated results have reaffirmed the view that correlated equilibria are computationally more tractable objects than Nash equilibria. As discussed, this is certainly the case in games of polynomial type; but is this viewpoint justified beyond that?

\subsection{Our results}

To cast light on this fundamental question, one of our key conceptual contributions is to connect correlated equilibria to a natural relaxation of fixed points, which we call \emph{fixed sets}. A fixed set of $f : \cX \to \cX$ is a nonempty finite set $S \subseteq \cX$ such that $S \subseteq \operatorname{conv} f(S)$ (\Cref{fig:fixed-sets}, right). A fixed set supported on a \emph{single point} is nothing other than a fixed point (\Cref{fig:fixed-sets}, left). Another special case is a \emph{cycle} or periodic orbit under $f$, whereby $S \subseteq f(S)$ (\Cref{fig:fixed-sets}, center). At the heart of our approach is the use of fixed sets as a stepping stone toward characterizing the complexity of correlated equilibria. (\Cref{def:fixed set} gives the approximate version used in our computational results.)

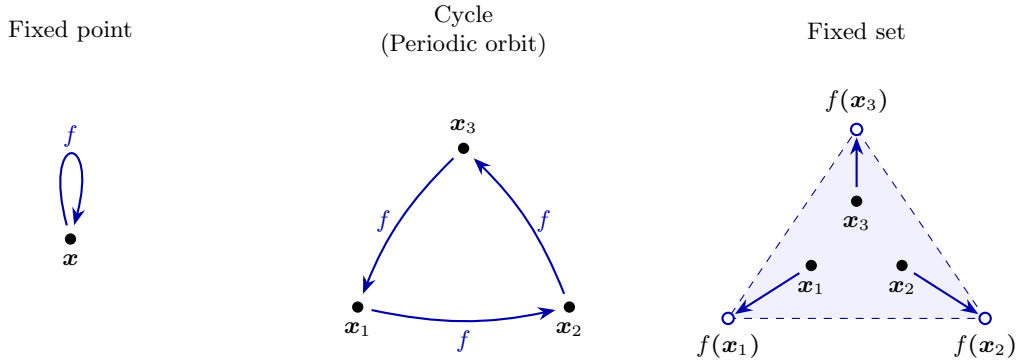
\begin{figure}[htbp]
\centering
\input{tikz/fixed-sets.tex}
\caption{A fixed point, a periodic-3 orbit, and a fixed set supported on 3 points.}
\label{fig:fixed-sets}
\end{figure}

Our main technical result establishes that computing a fixed set is as hard as computing a fixed point.

\begin{restatable}[Hardness of fixed sets]{maintheorem}{fixedsethardness}
    \label{th:fixed set hardness}
    For every fixed $C\ge1$, there is a constant $c>0$ such that the following problem is \PPAD-hard: given a Lipschitz function $f:[0,1]^d\to[0,1]^d$ that can be evaluated in polynomial time, compute a $d^{-c}$-fixed set $S$ of $f$ with $|S|\le d^C$.
\end{restatable}

It is worth pointing out that this hardness result, and subsequent implications thereof, extend---under randomized reductions---to distributions with superpolynomial support of size at most $K$, provided the output admits polynomial-time sampling and $\log K$ is polynomial in the input size (\Cref{rem:sampling representations}). Also, the tradeoff between precision and support in~\Cref{th:fixed set hardness} is essential: there is a simple FPTAS for computing an $\eps$-fixed set based on iterating $f$, supported on $O(1/\eps)$ points (\Cref{rem:fixed-set-fptas}).

As alluded to, the connection we make between fixed sets and correlated equilibria yields the following direct consequence for quadratic concave games.\footnote{Sometimes the term ``normal-form correlated equilibrium'' is used to distinguish from other correlated equilibrium concepts, but we will not do so here.}

\begin{maintheorem}[CE hardness in concave games]\label{cor:concave hardness}
    For every fixed $C\ge1$, there is a constant $c>0$ such that computing a $d^{-c}$-correlated equilibrium supported on at most $d^C$ strategy profiles is \PPAD-hard, even in two-player concave games with $\cX_1 = \cX_2 = [0, 1]^d$ where each utility is quadratic in the player's own strategy.
\end{maintheorem}

That is, surprisingly, computing a correlated equilibrium is as hard as computing a Nash equilibrium, challenging the conventional view in game theory. \Cref{cor:concave hardness} follows from~\Cref{th:fixed set hardness} using the standard two-player \emph{imitation game} (\emph{e.g.},~\citealp{Mclennan10:Simple,Babichenko16:Query,Rubinstein16:Settling}), with utilities $u_1(\vx,\vy)=-\norm{\vx-\vy}_2^2$ and $u_2(\vx,\vy)=-\norm{\vy-f(\vx)}_2^2$ (\Cref{prop:ce-fixed-set reduction}).

Moreover, the same reduction readily gives an unconditional lower bound in the query model.

\begin{maintheorem}[Query hardness]\label{cor:concave query hardness}
    For every fixed $C\ge 1$, there is a constant $c>0$ such that any randomized algorithm that, with probability at least $2/3$, computes a $d^{-c}$-correlated equilibrium supported on at most $d^C$ profiles requires $2^{\Omega(d/\log d)}$ utility queries in the worst case. This holds even in two-player concave games on $[0,1]^d$, where each utility is quadratic in the player's own strategy.
\end{maintheorem}

To see this, we start from the $2^{\Omega(d_0)}$ randomized query lower bound for finding a fixed point of a Lipschitz map on $[0,1]^{d_0}$, which holds even when the Lipschitz and precision parameters are absolute constants~\citep{Hirsch89:Exponential,Rubinstein16:Settling}. For a prescribed polynomial support bound, our reduction constructs a game of dimension $d=O(d_0\log d_0)$ in which each utility query can be simulated with $O(\log d_0)$ queries to the original map.

Finally, we obtain hardness for multilinear games on explicitly represented polytopes.
\begin{restatable}[CE hardness in multilinear polytope games]{maintheorem}{polytopemult}
\label{th:multilinear hardness}
    Let $N$ be the binary encoding length of the game's description. For every fixed $C \ge 1$, computing an $\eps$-correlated equilibrium supported on at most $N^C$ profiles in a multilinear game on explicit rational polytopes is
    \begin{enumerate}[label=(\roman*)]
        \item \PPAD-hard at accuracy $\eps=2^{-\operatorname{poly}(N)}$;
        \item \PPAD-hard at accuracy $\eps=N^{-c}$ for some constant $c>0$, assuming the PCP-for-PPAD conjecture (\Cref{ass:pcp-ppad}).\footnote{OpenAI recently claimed a proof of the required \href{https://github.com/openai/math/blob/main/preprints/The-PCP-for-PPAD-conjecture-a-quasilinear-reduction-September-25-2026/paper.pdf}{PCP-for-PPAD theorem}. We refrain from using this claim as an established theorem until it has been verified by human experts, and state this part conditionally.}
    \end{enumerate}
    Both statements hold even when the number of players is an absolute constant.
\end{restatable}

It is interesting to point out that our results immediately also imply \PPAD-hardness for computing a point in the convex hull of Nash equilibrium distributions---which is another problem we have been curious about---under the same output representation and support-accuracy tradeoff. Indeed, every convex combination of Nash equilibrium distributions is a correlated equilibrium.

\subsection{Technical overview}

We now give an informal sketch of the proof of \Cref{th:fixed set hardness}, setting aside issues such as approximation. The quantitative version of the argument appears in~\Cref{sec:hardness-fixedsets}.

It is instructive to first treat an intermediate case between fixed points and fixed sets---periodic orbits. In this special case, hardness follows directly from a scaling argument, which has the effect of \emph{amplifying the length of any periodic orbit}. In particular, given an $L$-Lipschitz map $f : [0,1]^d \to [0,1]^d$, we define $f_\alpha(\vx)=(1-\alpha)\vx+\alpha f(\vx)$ for a small parameter $\alpha \in (0,1]$. If $\vx_1,\dots,\vx_k$ form a periodic orbit of $f_\alpha$, with $\vx_{k+1}=\vx_1$, then consecutive points are at distance at most $\alpha$ in the normalized Euclidean norm, so the orbit has diameter at most $k\alpha$. Moreover, summing around the orbit gives $\sum_{i=1}^k (f(\vx_i)-\vx_i)=0$, and hence
\begin{align}
    \norm{f(\vx_1)-\vx_1}_2
    &\le \frac1k\sum_{i=1}^k \norm{(f(\vx_1)-\vx_1)-(f(\vx_i)-\vx_i)}_2
    \le (L+1)k\alpha.
\end{align}
Thus, choosing $\alpha\le\delta/((L+1)k)$ ensures that any orbit of length at most $k$ yields a $\delta$-fixed point of $f$.

However, the challenge for a general fixed set is that each point is a convex combination of images of other points, rather than the image of a single predecessor. Instead of using a single scale, the key idea is to use \emph{several scales} simultaneously: the same transition matrix must form a fixed set at every scale, and we exploit this constraint to recover an approximate fixed point.

Specifically, we construct a function $g : [0, 1]^{dr} \to [0, 1]^{dr}$, where $r$ will be picked later, as follows: we pick $\alpha_1, \dots, \alpha_r \in (0,1]$, and set
\begin{align}
    g(\vx^{(1)}, \dots, \vx^{(r)}) = \ab(f_{\alpha_1}(\vx^{(1)}), \dots, f_{\alpha_r}(\vx^{(r)})).
\end{align}
(This lifting is illustrated in~\Cref{fig:fixed-sets-scales}.) Let $S = \{ \vx_1, \dots, \vx_k \}$ be a fixed set of $g$, and let $P$ be a transition matrix satisfying $\vx_i = \E_{j \sim P(i)} g(\vx_j)$ for every $i \in [k]$. Our goal is to show that $S$ yields an approximate fixed point of $f$; indeed, we will show that there are $i \in [k]$ and $\ell \in [r]$ for which $\vx^{(\ell)}_i$ is an approximate fixed point of $f$.
\begin{figure}[!htb]
\centering
\includegraphics[width=\textwidth]{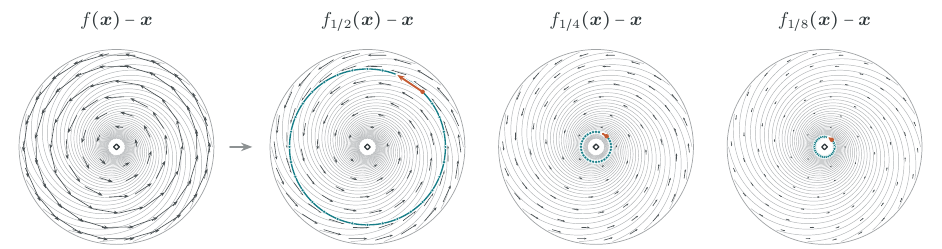}\par
\nointerlineskip
\vspace{4pt}
{\color{black!20}\rule{0.96\textwidth}{0.3pt}}\par
\nointerlineskip
\vspace{4pt}
\includegraphics[width=\textwidth]{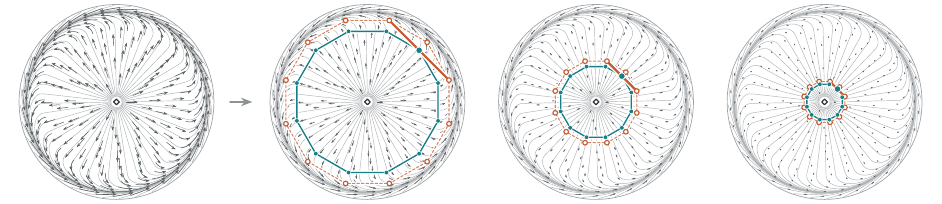}
\caption[Fixed sets at three scales of two different maps.]{Top row. The original rotation $f$ and three scaled maps $f_\alpha$, with common fixed point $\vx_*=0$. Gray curves depict the displacement flow, and black arrows show $f_\alpha(\vx)-\vx$. Teal points form approximate fixed sets with the same transition matrix and equal error across the three scales.
\par\smallskip
Bottom row. A different map $f$ and its three scaled maps. Teal points form regular $12$-gons; hollow orange points are their images under $f_\alpha$. Each support point is the midpoint of the images of its two neighbors, as illustrated by the orange chord. Thus, $P(i,i-1)=P(i,i+1)=1/2$ works at all three scales, with uniform stationary distribution.}
\label{fig:fixed-sets-scales}
\end{figure}

The key observation is the following. For $\ell \in [r]$ let $S_\ell = \{ \vx_1^{(\ell)}, \dots, \vx_k^{(\ell)} \}$. Then, by definition, $S_\ell$ is a fixed set of $f_{\alpha_\ell}$ with transition matrix $P$. Therefore, we have constructed fixed sets at multiple ``scales'' of $f$, {\em that use the same transition matrix}. We show that if we do this for ``enough'' values of $\alpha \in (0,1]$, then we must find a fixed point.

\begin{proposition}\label{prop:non-quantitative}
    Let $P : [k] \to \Delta(k)$ be a transition matrix. For every $\alpha \in (0,1]$ let $S_\alpha$ be a fixed set of $f_\alpha$ with transition matrix $P$. Then the closure of $\bigcup_{\alpha \in (0,1]} S_\alpha$ contains a fixed point of $f$.
\end{proposition}
\begin{proof}
    Assume WLOG that the transition graph of $P$ is strongly connected; otherwise, replace it with a sink strongly connected component. Let $S_\alpha = \{ \vx_1^{(\alpha)}, \dots, \vx_k^{(\alpha)} \}$. Then we have
    \begin{align}\label{eq:non-quantitative}
        \vx_i^{(\alpha)} = \E_{j \sim P(i)} \ab[ (1 - \alpha) \vx_j^{(\alpha)} + \alpha f(\vx_j^{(\alpha)})].
    \end{align}
    Let $\{\alpha_r\}_{r=1}^\infty$ be a sequence converging to $0$, and $\vx_i = \lim_{r \to \infty} \vx_i^{(\alpha_r)}$. (If this limit does not exist, replace $\{\alpha_r\}$ by a subsequence so that it does; this is allowed because of compactness.) Replacing this limit in \eqref{eq:non-quantitative} yields
    \begin{align}
        \vx_i = \E_{j \sim P(i)} \vx_j.
    \end{align}
    We claim now that $\vx_1 = \dots= \vx_k$. Indeed, consider any direction $\va$ and let $i$ be chosen to maximize $\ip{\va, \vx_i}$. Then
    \begin{align}
        \ip{\va, \vx_i} = \E_{j \sim P(i)} \ip{\va, \vx_j}.
    \end{align}
    By induction and strong connectivity of the transition graph, it follows that $\ip{\va, \vx_j} = \ip{\va, \vx_i}$ for {\em every} $j$. Since this is true for every $\va$, we have  $\vx_1 = \dots= \vx_k \eqdef \vx$.
    
    Moreover, taking expectations in \eqref{eq:non-quantitative} over the stationary distribution $\mu$, and dividing by $\alpha > 0$, yields
    \begin{align}\label{eq:qualitative-efp}
        \E_{i \sim \mu} \vx_i^{(\alpha)} = \E_{i \sim \mu}f(\vx_i^{(\alpha)}).
    \end{align}
    Since $\vx_i^{(\alpha_r)} \to \vx$ for all $i$, taking the limit $r \to \infty$ now yields $\vx = f(\vx)$.
\end{proof}

\paragraph{The multilinear extension.} We now explain how to go from concave quadratic utilities to a multilinear polytope game. The obvious approach that comes to mind here is to employ some form of discretization. However, the issue is that this would make the game description grow polynomially with the inverse accuracy. At the same time, \Cref{th:fixed set hardness} prescribes that the accuracy needs to be small enough relative to the size of the support. As a result, it is not clear whether there is a discretization that can handle any polynomial support bound. 

Instead, to arrive at~\Cref{th:multilinear hardness}, the key step is to essentially transfer the quadratic imitation term from the utility into the strategy constraints. Indeed, since $-\norm{\vx-\vy}_2^2 = \frac{2}{d}\ip{\vx,\vy}-\norm{\vx}_2^2-\norm{\vy}_2^2$ (in relative $\ell_2$ norm), it suffices to incorporate an auxiliary variable whose value approximates $\norm{\vx}_2^2$, for $\vx$; the last term does not depend on $\vx$, while the interaction term is already bilinear. The challenge is to express these auxiliary variables with few linear constraints. To do so, we rely on an approach in approximation theory (\emph{e.g.,}~\citealp{Yarotsky17:Error}), whereby each term $x_i^2$ in $\|\vx\|_2^2$ is replaced by a piecewise-linear interpolation that can be described with only $O(m)$ auxiliary variables and inequalities.

To cast the remaining polynomial interactions into multilinear form, we use the standard trick of assigning different factors of each monomial to different ``copy players.'' In particular, we can start with a polynomial map $f:[0,1]^d \to[0,1]^d$ of constant degree whose approximate fixed points are \PPAD-hard to compute. Since $f$ has constant degree, we require only a constant number of players. These steps are described at greater length in~\Cref{sec:polytope}.

\subsection{Further related work}

We now elaborate on related threads of research.

\paragraph{Computing correlated equilibria beyond polynomial type.} For games of polynomial type, as discussed, there are polynomial-time algorithms computing exact CE~\citep{Papadimitriou08:Computing,Jiang15:Polynomial-Time}. The closest predecessor to our work is~\citet{Anagnostides26:Complexity}, who introduced convex expected fixed points and established hardness of computing CE in concave quadratic games relative to the problem of finding a fixed point of a contraction~\citep{Fearnley20:Unique}. We strengthen this computational barrier to \PPAD-hardness, and further extend it to multilinear games on explicitly represented polytopes. A complementary positive result was obtained by~\citet{Anagnostides26:Normal}, who showed that CEs can be computed efficiently in finite-horizon Markov games with a fixed number of players. Related hardness results concern \emph{sparse} coarse correlated equilibria (CCE), represented as mixtures of a small number of product distributions: \citet{Foster23:Hardness} established hardness in Markov games, yielding lower bounds for independent learning. Analogous hardness results were later established in extensive-form games~\citep{Anagnostides24:Complexity,Peng24:Complexity}. Unlike these latter two results for extensive-form games, which apply even to coarse correlated equilibria, our main result rules out efficient computation for \emph{any} prescribed polynomial support bound. The complexity of computing an \emph{optimal} correlated equilibrium---as opposed to computing \emph{any} correlated equilibrium, which is our focus in this paper---in extensive-form games was recently characterized by~\citet{Cheval26:Complexity}.

\paragraph{Tractable relaxations.} One fruitful research direction has been to obtain tractable equilibrium concepts by restricting the set of deviations, in extensive-form games~\citep{Stengel08:Extensive} and beyond. \citet{Farina24:Polynomial} gave a polynomial-time algorithm for \emph{linear $\Phi$-equilibria}---whereby only linear deviations are considered---in polyhedral games, while~\citet{Daskalakis24:Efficient} established efficient learning and computation in general convex games. \citet{Arunachaleswaran25:Swap} introduced \emph{profile swap regret} for polytope games and gave efficient algorithms guaranteeing sublinear profile swap regret, which they showed is necessary and sufficient for \emph{non-manipulability}. This extends the non-manipulability guarantee of \emph{swap regret} in normal-form games~\citep{Deng19:Strategizing}, where a strategic opponent cannot asymptotically obtain more than its Stackelberg payoff against a learner with sublinear swap regret; we recall that sublinear swap regret leads to correlated equilibria~\citep{Hart00:Simple,Blum07:From}. \citet{Zhang25:Learning} considered a class of nonlinear deviations, including low-degree polynomial deviations. These developments build on the connection between fixed-point computation and regret minimization~\citep{Hazan07:Computational,Gordon08:No}. A central ingredient in their approach is the notion of an \emph{expected} fixed point, which can be computed efficiently~\citep{Zhang25:Learning}. Our fixed sets correspond instead to \emph{convex} expected fixed points, revealing a computational separation between these two relaxations of fixed points.

\paragraph{Learning with large strategy spaces.}
\citet{Stoltz07:Learning} examined convergence to correlated equilibria in games with convex strategy sets, and~\citet{Stein07:Characterization} developed approximation algorithms for correlated equilibria in continuous games, including polynomial games. More recently,~\citet{Peng24:Fast,Dagan24:From} obtained algorithms for minimizing swap regret even in large strategy spaces, yielding a polynomial-time approximation scheme for correlated equilibria in, for example, extensive-form games. However, the learning hardness result of~\citet{Anagnostides26:Complexity} rules out swap-regret guarantees polynomial in the game size and inverse accuracy. Such lower bounds leave open the possibility of efficient equilibrium computation, which is the focus of our paper. On the technical side, our entropy argument is related in spirit to the information-theoretic approach of~\citet{Harrow16:Tight} for approximating Nash equilibria.

\paragraph{Calibration and economic motivation.} Correlated equilibria are also intimately connected to forecasting and \emph{calibration}. \citet{Foster97:Calibrated} showed that best responses to calibrated forecasts of opponents' play lead the empirical distribution toward correlated equilibria. This connection has been further cultivated in recent years, where improved algorithms for minimizing swap regret~\citep{Dagan24:From,Peng24:Fast} have led to efficient algorithms for high-dimensional calibration~\citep{Peng25:Calibration,Fishelson25:Calibration}. Another use of swap regret is in the context of \emph{algorithmic collusion}~\citep{Calvano20:Artificial}, where it can be used to mitigate supracompetitive pricing~\citep{Hartline24:Regulation,Hartline25:Regulation}. Indeed, vanishing external regret alone can still permit collusive behavior. Recently, \citet{Hartline26:Computable} introduced the notion of \emph{information-value-free correlated equilibria}, and showed that they can be computed efficiently offline but cannot be learned by a broad class of learning algorithms. These connections show the importance of understanding the computational complexity of equilibrium concepts stronger than coarse correlated equilibria, which is what motivates our paper.

%% file: tikz/fixed-sets.tex
\begin{tikzpicture}[
  font=\small,
  >=Stealth,
  point/.style={circle,fill=black,inner sep=1.5pt},
  image point/.style={circle,draw=blue!65!black,fill=white,thick,inner sep=1.5pt},
  map arrow/.style={->,thick,blue!65!black,shorten >=3pt,shorten <=3pt}
]
  \begin{scope}
    \node[] at (2,4.25) {Fixed point};
    \node[point,label=below:{$\vx$}] (point) at (2,1.5) {};
    \draw[map arrow] (point) edge[loop above,min distance=14mm]
      node[above] {$f$} (point);
  \end{scope}

  \begin{scope}[xshift=5.2cm]
    \node[align=center] at (2,4.25)
      {Cycle\\(Periodic orbit)};
    \node[point,label=below:{$\vx_1$}] (cycle1) at (0.6,0.6) {};
    \node[point,label=below:{$\vx_2$}] (cycle2) at (3.4,0.6) {};
    \node[point,label=above:{$\vx_3$}] (cycle3) at (2,2.7) {};
    \draw[map arrow] (cycle1) to[bend right=12] node[below] {$f$} (cycle2);
    \draw[map arrow] (cycle2) to[bend right=12] node[right] {$f$} (cycle3);
    \draw[map arrow] (cycle3) to[bend right=12] node[left] {$f$} (cycle1);
  \end{scope}

  \begin{scope}[xshift=10.4cm]
    \node[align=center] at (2,4.25)
      {Fixed set};
    \coordinate (image1) at (0.3,0.45);
    \coordinate (image2) at (3.7,0.45);
    \coordinate (image3) at (2,2.95);
    \filldraw[fill=blue!6,draw=blue!55!black,dashed]
      (image1) -- (image2) -- (image3) -- cycle;
    \node[point,label=below:{$\vx_1$}] (set1) at (1.4,1.15) {};
    \node[point,label=below:{$\vx_2$}] (set2) at (2.6,1.15) {};
    \node[point,label=below:{$\vx_3$}] (set3) at (2,2) {};
    \draw[map arrow] (set1) -- (image1);
    \draw[map arrow] (set2) -- (image2);
    \draw[map arrow] (set3) -- (image3);
    \node[image point,label=below:{$f(\vx_1)$}] at (image1) {};
    \node[image point,label=below:{$f(\vx_2)$}] at (image2) {};
    \node[image point,label=above:{$f(\vx_3)$}] at (image3) {};
  \end{scope}
\end{tikzpicture}

%% file: text/preliminaries.tex
\section{Preliminaries}

We work in the standard Turing model. All numbers are rational numbers encoded in binary. Functions are given as arithmetic circuits with standard gates, see \eg,~\citet{Rubinstein16:Settling}. These choices are not crucial to our results; indeed, our results will follow from black-box reductions from Brouwer fixed points to our problems of interest.

Following~\citet{Rubinstein16:Settling}, for $\vx\in\R^d$ and $1\le p<\infty$ we adopt the normalized norm convention
\begin{align}
    \norm{\vx}_p = \ab(\frac{1}{d} \sum_{i=1}^d |x_i|^p)^{1/p}.
\end{align}
Inner products are defined as usual without any normalization: $\ip{\vx,\vy}=\sum_{i=1}^d x_i y_i$, so $|\ip{\vx,\vy}|/d\le\norm{\vx}_2\norm{\vy}_2$.

\begin{definition}\label{def:fixed point}
    Let $f : \cX \to \cX$, where $\cX \subseteq \R^d$. An $\eps$-{\em fixed point} of $f$ is a point $\vx \in \cX$ such that
    \begin{align}
        \norm{f(\vx) - \vx}_2 \le \eps,
    \end{align}
    where $\norm{\cdot}_2$ is the normalized Euclidean norm. A $0$-fixed point is simply called a fixed point.
\end{definition}

To avoid measure-theoretic issues, we assume throughout that distributions have finite support. Algorithms that output distributions will give outputs by explicitly listing all the support points and their associated probabilities. A polynomial-time algorithm can therefore output only distributions supported on polynomially many points, with probabilities of polynomial bit length.

\input{text/concave-games.tex}

\subsection{The PCP-for-PPAD conjecture}

We now state the precise version of the \emph{PCP-for-PPAD conjecture} of~\citet{Babichenko16:Can} we need for~\Cref{th:multilinear hardness}. A \emph{generalized circuit} assigns a value in $[0,1]$ to each of $d$ nodes, each of which is the output of exactly one gate. Gates have at most two inputs and cycles are allowed. The circuit has \emph{fan-out} two if each node occurs in at most two input positions. All gate parameters $\beta$ are rational numbers in $[0,1]$. For input values $a,b$ and output value $c$, the arithmetic gates are satisfied at tolerance $\tau$ when the following conditions hold depending on the type of the gate:
\begin{itemize}
    \item \emph{constant}: $|c-\beta|\le\tau$;
    \item \emph{scaling}: $|c-\beta a|\le\tau$;
    \item \emph{copying}: $|c-a|\le\tau$;
    \item \emph{addition}: $|c-\min\{a+b,1\}|\le\tau$; and
    \item \emph{subtraction}: $|c-\max\{a-b,0\}|\le\tau$.
\end{itemize}
Finally, the comparison gate is satisfied when both implications
\begin{align}
    a<b-\tau\ \Longrightarrow\ c\ge1-\tau
    \text{ and } a>b+\tau\ \Longrightarrow\ c\le\tau
\end{align}
hold. We are ready to state the formal version of the conjecture.

\begin{assumption}[PCP-for-PPAD]\label{ass:pcp-ppad}
    There are absolute constants $0<\tau<1/4$ and $0<\gamma<1$ such that the following problem is \PPAD-hard: given a generalized circuit with fan-out two, compute an assignment satisfying at least $(1-\gamma)d$ gates at tolerance $\tau$. 
\end{assumption}

That is, this postulates that \PPAD-hardness of generalized circuits persists even when an arbitrary set of at most $\gamma d$ gates can fail.

%% file: text/concave-games.tex
\subsection{Concave games and correlated equilibria}

An $n$-player {\em concave game} $\Gamma$ is given by convex compact strategy sets $\cX_i \subset \R^{d_i}$ and utilities $u_i : \cX \to \R$, for $i\in[n]$, where $\cX=\prod_{i=1}^n\cX_i$. For each player $i$, $u_i(\vx_i,\vx_{-i})$ is concave in $\vx_i$ for every fixed profile $\vx_{-i}$ of the other players' strategies.

A \emph{multilinear polytope game} is a special case in which each $\cX_i$ is a compact polytope and each utility is a multilinear polynomial, so that every monomial contains at most one coordinate of each player's strategy (raised to the first power). Thus, each utility is affine in a player's strategy when the other players' strategies are fixed. We represent each polytope by an explicit list of linear inequalities with rational coefficients and each utility by an explicit list of monomials with rational coefficients. We write $N$ for the total encoding length of this game description.

\begin{definition}[\citealp{Aumann74:Subjectivity}]
An $\eps$-{\em correlated equilibrium ($\eps$-CE)} of a game $\Gamma$ is a distribution $\mu \in \Delta(\cX)$ such that
\begin{align}
    \E_{\vx \sim \mu} \ab[ u_i(\phi_i(\vx_i), \vx_{-i}) - u_i(\vx) ] \le \eps
\end{align}
for every player $i\in[n]$ and function $\phi_i : \cX_i \to \cX_i$.    
\end{definition}

%% file: text/fixed-sets.tex
\section{Fixed sets and convex expected fixed points}

In this section, we introduce the notion of a \emph{fixed set}, and connect it to an earlier definition by~\citet{Anagnostides26:Complexity}. Let $\cX \subset \R^d$ be convex and compact, and $f : \cX \to \cX$ be a continuous function. 

\begin{definition}[\citealp{Anagnostides26:Complexity}] \label{def:cefp}
An $\eps$-{\em convex expected fixed point} (hereafter {\em convex EFP}) of $f$ is a distribution $\mu \in \Delta(\cX)$ such that
\begin{align}
    \E_{\vx\sim\mu} \ab[Q(\vx) - Q(f(\vx))] \le \eps
\end{align}
for every convex $Q : \cX \to \R$ that is $1$-Lipschitz under $\norm{\cdot}_2$.
\end{definition}

Restricting $Q$ to linear $1$-Lipschitz functions gives the condition $\norm{\E_{\vx\sim\mu}[\vx-f(\vx)]}_2\le\eps$, known simply as an \emph{expected fixed point}~\citep{Zhang24:Efficient}, which was shown to admit a polynomial-time algorithm~\citep{Zhang25:Learning}.

\begin{definition}\label{def:fixed set}
    An $\eps$-{\em fixed set} of $f$ is a finite set $S = \{ \vx_1, \dots, \vx_k \} \subset \cX$ such that there is some Markov chain transition function $P : [k] \to \Delta([k])$, and stationary distribution $\mu$ of $P$, such that
\begin{align}
    \E_{i \sim \mu} \norm*{\vx_i - \E_{j \sim P(i)} f(\vx_j)}_2 \le \eps.
\end{align}
\end{definition}

Using duality, it turns out that the previous two definitions are equivalent.

\begin{theorem}\label{th:fixed set equivalence}
    The above two definitions are equivalent. That is, the support of an $\eps$-convex EFP is also an $\eps$-fixed set; conversely, given an $\eps$-fixed set, the corresponding stationary distribution $\mu$ is an $\eps$-convex EFP.
\end{theorem}

\begin{proof}
    Let $S$ be an $\eps$-fixed set, $P$ its corresponding transition matrix, and $\mu$ be its stationary distribution. Then for any convex $1$-Lipschitz $Q$, we have
    \begin{align}
        \E_{\vx} \ab[Q(\vx) - Q(f(\vx))] &= \E_\vx Q(\vx) - \E_\vx \E_{\vx'\mid\vx} Q(f(\vx'))
    \le \E_\vx \ab[Q(\vx) - Q\ab\Big(\E_{\vx'\mid\vx}  f(\vx'))]
        \le \E_\vx \norm[\Big]{\vx - \E_{\vx'\mid\vx} f(\vx')}_2 \le \eps,
    \end{align}
    where $\vx\sim\mu$ and $\vx'\mid\vx$ means $\vx'\sim P_\vx$. The first equality uses stationarity of $\mu$ under $P$. The inequalities are, in order, Jensen and Lipschitzness of $Q$.

    Conversely, let $\mu$ be an $\eps$-convex EFP with support $\{\vx_1, \dots, \vx_k\}$. We need to show that there is a transition matrix $P$ satisfying the fixed set condition. To do this, consider the following convex program.
    \begin{align}
        \min_{\mW \ge \vec 0} \quad \sum_i \norm*{\sum_j W_{ij} (\vx_i - f(\vx_j))}_2 \qq{s.t.} \mW^\top \vec 1 = \mW \vec 1 = \mu.
    \end{align}
    It suffices to show that the value of this program is $\le \eps$, because then it would be enough to use the transition matrix $P(j\mid i) = W_{ij}/\mu_i$.
    Using $\norm{\vv}_2=\max_{\norm{\va}_2\le1}\ip{\va,\vv}/d$, the dual of this problem is
    \begin{align}
        \max \E_{i \sim \mu} [b_i - c_i] \qq{s.t.} \frac1d\ip*{\va_i,\vx_i - f(\vx_j)} \ge b_i - c_j, \quad \norm{\va_i}_2 \le 1.
    \end{align}
    For any dual-feasible solution, we consider the function $$Q(\vx) = \max_i \{ b_i + \tfrac1d\ip{\va_i, \vx - \vx_i} \}.$$
    Then $Q$ is $1$-Lipschitz and convex. Moreover, we have $Q(\vx_i) \ge b_i$ and $Q(f(\vx_i)) \le c_i$,
    so
    \begin{align}
        \E_{i \sim \mu} [b_i - c_i] \le \E_{i \sim \mu} \ab[Q(\vx_i) - Q(f(\vx_i))] \le \eps. \tag*\qedhere
    \end{align}
\end{proof}

\begin{remark}
    An alternative way to see this equivalence is through Strassen's theorem~\citep{Gozlan17:Kantorovich}, which asserts that for probability measures $\mu,\nu$ on $\R^d$ with finite first moments, $\int Q\,d\mu\le\int Q\,d\nu$ for every convex and Lipschitz $Q$ if and only if they admit a coupling $(X,Y)$ with $\E[Y\mid X]=X$. The approximate version~\citep[Remark 3.4]{Gozlan17:Kantorovich} states that
    \begin{align}
        \int Q\,d\mu-\int Q\,d\nu\le\eps
        \quad\text{for every convex $1$-Lipschitz $Q$}
    \end{align}
    if and only if there is a coupling $X\sim\mu$, $Y\sim\nu$ satisfying
    \begin{align}
        \E\norm{X-\E[Y \mid X]}_2\le\eps.
    \end{align}
    To apply this to \Cref{th:fixed set equivalence}, let $\nu$ be the law of $f(X)$ for $X\sim\mu$. The convex EFP condition gives the coupling above. Lift $Y$ to $X'\sim\mu$ by sampling from $\mu$ conditional on $f(X')=Y$, independently of $X$ given $Y$. Then $f(X')=Y$, and $P(\vx'\mid\vx)=\Pr(X'=\vx'\mid X=\vx)$ has stationary distribution $\mu$, since both $X$ and $X'$ have law $\mu$. The coupling inequality is precisely the $\eps$-fixed-set condition. The converse follows directly from Jensen's inequality and Lipschitzness.
\end{remark}

\begin{remark}[Contractions]
    If $f$ is $L$-Lipschitz for $L<1$, with unique fixed point $\vx_*$, then every $\eps$-fixed set satisfies
    \begin{align}
        \E_{i \sim \mu} \norm{\vx_i - \vx_*}_2 \le \frac{\eps}{1-L}.
    \end{align}
    Indeed, applying the convex EFP condition to $Q(\vx) = \norm{\vx - \vx_*}_2$ gives
    \begin{align}
        (1-L)\E_{i \sim \mu} \norm{\vx_i - \vx_*}_2
        \le \E_{i \sim \mu}\ab[Q(\vx_i) - Q(f(\vx_i))] \le \eps.
    \end{align}
    In particular, for an exact fixed set, the stationary distribution $\mu$ is concentrated on $\vx_*$.
\end{remark}

\subsection{Correlated equilibria and fixed sets}

The following reduction connects correlated equilibria to fixed sets. It follows the proof of~\citet[Theorem~3.1]{Anagnostides26:Complexity}, but we include the argument for completeness. 

\begin{proposition}[Reducing fixed sets to correlated equilibria]
\label{prop:ce-fixed-set reduction}
    Given a map $f:[0,1]^d\to[0,1]^d$, we construct the two-player concave game on $\cX_1=\cX_2=[0,1]^d$ with utilities
    \begin{align}
        u_1(\vx,\vy)=-\norm{\vx-\vy}_2^2, \quad
        u_2(\vx,\vy)=-\norm{\vy-f(\vx)}_2^2.
    \end{align}
    For every $\eps\ge0$ and every $\eps$-CE $\mu$ of this game, $S=\supp\mu_\vx$ is a $2\sqrt{\eps}$-fixed set of $f$ with $|S|\le|\supp\mu|$, with a transition matrix whose stationary distribution is $\mu_\vx$.
\end{proposition}

\begin{proof}
    Each utility is concave and quadratic in its player's own strategy. The CE conditions imply
    \begin{align}\label{eq:ce-condition}
        \E_\vx \norm*{\vx - \E_{\vy\mid\vx} \vy}_2^2 \le \eps, \quad \E_\vy \norm*{\vy - \E_{\vx\mid\vy} f(\vx)}_2^2 \le \eps,
    \end{align}
    where $\E_\vx$ and $\E_\vy$ denote sampling from the marginals of $\mu$, and $\E_{\vx\mid\vy}$ and $\E_{\vy\mid\vx}$ denote sampling from the conditional distributions of $\mu$ given $\vy$ and $\vx$, respectively. Then
    \begin{align}
        \E_\vx \norm[\Big]{\vx - \E_{\vy\mid\vx, \vx'\mid\vy} f(\vx')}_2
        &\le \E_{\vx} \ab[\norm[\Big]{\vx - \E_{\vy\mid\vx} \vy}_2 + \norm[\Big]{\E_{\vy\mid\vx} \ab\big[\vy - \E_{\vx'\mid\vy} f(\vx')]}_2 ]
        \\&\le \E_\vx \ab[\norm[\Big]{\vx - \E_{\vy\mid\vx} \vy}_2 + \E_{\vy\mid\vx} \norm[\Big]{ \vy - \E_{\vx'\mid\vy} f(\vx')}_2]
        \\&\le \ab(\sqrt{\E_\vx \ab[\norm[\Big]{\vx - \E_{\vy\mid\vx} \vy}_2^2]} + \sqrt{\E_{\vy} \ab[\norm[\Big]{ \vy - \E_{\vx'\mid\vy} f(\vx')}_2^2]})
        \\&\le 2 \sqrt{\eps},
    \end{align}
    where we use, in order, the triangle inequality, Jensen, Cauchy-Schwarz, and finally \eqref{eq:ce-condition}.
    Set $S=\supp\mu_\vx$ and $P(\vx'\mid\vx)=\sum_\vy\mu(\vy\mid\vx)\mu(\vx'\mid\vy)$. This transition matrix has stationary distribution $\mu_\vx$, since
    \begin{align}
        \sum_\vx\mu_\vx(\vx)P(\vx'\mid\vx)
        =\sum_\vy\mu_\vy(\vy)\mu(\vx'\mid\vy)
        =\mu_\vx(\vx').
    \end{align}
    This is precisely the $2\sqrt{\eps}$-fixed-set condition.
\end{proof}

%% file: text/amplification.tex
\section{Hardness of computing fixed sets}
\label{sec:hardness-fixedsets}

Let $f : [0, 1]^d \to [0, 1]^d$ be $L$-Lipschitz. It is known that computing a $\delta$-fixed point of $f$ is \PPAD-hard, even when $L\ge1$ and $0<\delta\le1$ are absolute constants~\citep{Rubinstein16:Settling}. Our goal is to prove the following.

\fixedsethardness*

The proof essentially makes the argument described in our technical overview quantitative. The remainder of this section is dedicated to its proof. To begin with, we construct the map 
\[
    g(\vx^{(1)}, \dots, \vx^{(r)}) = \ab(f_{\alpha_1}(\vx^{(1)}), \dots, f_{\alpha_r}(\vx^{(r)})),
\]
where the parameters $\alpha_1, \dots, \alpha_r$ will be chosen later. We let $S = \{ \vx_1, \dots, \vx_k \}$ be an $\eps$-fixed set of $g$ with $k\le K$, for a support bound $K \geq 2$ and precision $\eps > 0$ to be chosen later, and let $P$ and $\mu$ be the corresponding transition matrix and stationary distribution. In lieu of passing to the limit, we will select a layer $\ell \in [r]$ at which $\vx^{(\ell)}$ will have the same function as the limit $\vx$ in the proof of \Cref{prop:non-quantitative}. Our goal is to show that there is an $i \in [k]$ and $\ell \in [r]$ such that $\vx_i^{(\ell)}$ is a $\delta$-fixed point of $f$.

For a parameter $0<t\le1$, we define
\begin{align}
    P(t) \defeq \E_{a \sim \Geo(t)} P^a = t \sum_{a=0}^\infty (1-t)^a P^a = t(\mI-(1-t)P)^{-1}, \quad h(t) \defeq \E_{i \sim \mu} H(P_i(t))
\end{align}
where $H$ is entropy with natural logarithms and $P(1)=\mI$.
\begin{lemma}
    $h$ is nonincreasing.
\end{lemma}
\begin{proof}
    For $0<t<s\le1$, the stochastic matrix
    \begin{align}
        \mK(t,s)=\frac ts \mI+\ab(1-\frac ts)P P(t)
    \end{align}
    satisfies $P(t)=\mK(t,s)P(s)=P(s)\mK(t,s)$ and has stationary distribution $\mu$. But then we have
    \begin{align}
        h(t) - h(s) &= \E_{i \sim \mu} H(P_i(t)) - \E_{j \sim \mu}H(P_j(s))
        \\&= \E_{i \sim \mu, j \sim \mK_i} \ab[H(P_i(t)) - H(P_j(s))] \notag
        \\&= \E_{i \sim \mu, j \sim \mK_i} \KL(P_j(s) \lVert P_i(t)) \ge 0, \label{eq:entropy}
    \end{align}
    where $\KL(\cdot \lVert \cdot)$ denotes the KL divergence.
\end{proof}
Moreover, $P_i(t)$ is a distribution on $k$ items, so $h(t)$ is bounded in $[0, \log k]$. Thus, if we pick values $1 \eqdef t_0 > t_1 > t_2 > \dots > t_r > 0$, there will be some value $\ell \le r$ at which $$h(t_{\ell}) - h(t_{\ell-1}) \le \frac{\log k}{r}.$$
The use of entropy here is crucial, since we will later need the fact that the above bound is logarithmic in $k$. We fix this value of $\ell$, and for the rest of the proof, let $\vx_i \defeq \vx^{(\ell)}_i, \alpha \defeq \alpha_\ell, t \defeq t_\ell$.

\paragraph{Controlling the conditional variance.} For any $s \in [t, 1]$, we define
\begin{align}
    \bar{\vx}_i(s) \defeq \E_{j \sim P_i(s)} \vx_j, \quad V(s) \defeq \E_{i \sim \mu, j \sim P_i(s)} \norm{\vx_j - \bar{\vx}_i(s)}_2^2.
\end{align}
The quantity $V(s)$ is the variance of $\vx_j$ under $j\sim P_i(s)$, averaged over $i\sim\mu$; our goal is to show that $V(t)$ is small, so that these conditional distributions have small variance on average. We write
\begin{align}
    \vx_j = \E_{a \sim P_j}[ \vx_a + \alpha \Delta_a] + \ve_j,
\end{align}
where $\Delta_j = f(\vx_j) - \vx_j$. For a vector $\vv\in\R^{dr}$ with $d$-dimensional blocks, the normalized norms satisfy $\norm{\vv^{(\ell)}}_2\le\sqrt r\,\norm{\vv}_2$, so the mean error in this layer is at most $\eta \defeq \sqrt r\,\eps$ if $\norm{\vv} \leq \epsilon$. Thus,
\[
  \E_{j \sim \mu} \norm{\ve_j}_2 \le \eta.  
\]
For any $i$, taking expectations over $P_i(s)$ gives
\begin{align}\label{eq:inverse-identity}
    \E_{j \sim P_i(s),a \sim P_j}[ \alpha \Delta_a + \ve_j] &= \E_{j \sim P_i(s), a \sim P_j} [\vx_j - \vx_a] = \ip{P_i(s)(\mI-P), \vx},
\end{align}
where we use the notation $\ip{\mu, \vx} \defeq \sum_{i=1}^k \mu_i \vx_i$.

For $s<1$, recall that $P(s) = s(\mI - (1-s)P)^{-1}$.
Rearranging, we have
\begin{align}\label{eq:Ps-identity}
    P(s)(\mI-P) = \frac{s}{1-s}(\mI - P(s))
\end{align}
Substituting into \eqref{eq:inverse-identity} gives
\begin{align}\label{eq:error}
    \E_{j \sim P_i(s),a \sim P_j}[ \alpha \Delta_a + \ve_j] = \frac{s}{1-s}(\vx_i - \ip{P_i(s), \vx}) = \frac{s}{1-s}(\vx_i - \bar{\vx}_i(s)).
\end{align}
By triangle inequality, and taking expectations over $\mu$, it follows that
\begin{align}
    \E_{i \sim \mu} \norm{\vx_i - \bar{\vx}_i(s)}_2 \le \frac{1-s}{s} \E_{i \sim \mu, j \sim P_i(s),a \sim P_j} \ab[ \alpha \norm{\Delta_a}_2 + \norm{\ve_j}_2] \le \frac{1}{s}(\alpha + \eta) \label{eq:expected error bound}
\end{align}
where in the last inequality we use the trivial bounds $\norm{\Delta_a}_2 \le 1$ and $s \le 1$, as well as the fact that $\mu$ is stationary under $P(s)$.

Since $\bar{\vx}_i(s)$ is the mean of $\vx_j$ over $P_i(s)$, it minimizes the mean-squared error in $V(s)$. Therefore, we have
\begin{align}
    V(s) &\le \E_{i \sim \mu, j \sim P_i(s)} \norm{\vx_j - \vx_i}_2^2
    \\&= \frac{2}{d}\E_{i \sim \mu, j \sim P_i(s)} \ip{\vx_i - \vx_j, \vx_i}
    \\&= \frac{2}{d}\E_{i \sim \mu} \ip{\vx_i - \bar{\vx}_i(s), \vx_i}
    \\&\le 2 \E_{i \sim \mu} \norm{\vx_i - \bar{\vx}_i(s)}_2
    \\&\le \frac{2}{s}(\alpha + \eta).
\end{align}
where the first equality uses the fact that $i$ and $j$ have the same distribution (because $\mu$ is stationary), the second uses the definition of $\bar{\vx}_i(s)$, and the final two inequalities use Cauchy-Schwarz with $\norm{\vx_i}_2 \le 1$ and \eqref{eq:expected error bound}.

Now set $\mK \defeq \mK(t, s)$. Then
\begin{align}\label{eq:K transfer}
    \bar{\vx}_i(t) = \E_{j \sim P_i(t)} \vx_j = \E_{a \sim \mK_i, j \sim P_a(s)} \vx_j = \E_{j \sim \mK_i} \bar{\vx}_j(s).
\end{align}
Therefore,
\begin{align}
    V(t)-V(s) &= \E_{i\sim\mu}\ab[\norm{\bar{\vx}_i(s)}_2^2-\norm{\bar{\vx}_i(t)}_2^2]
    \\&= \E_{i \sim \mu, j \sim \mK_i} \norm{\bar{\vx}_j(s) - \bar{\vx}_i(t)}_2^2
    \\&= \E_{i \sim \mu, j \sim \mK_i} \norm*{\sum_a (P_{ja}(s) - P_{ia}(t)) \vx_a}_2^2
    \\&\le \E_{i \sim \mu, j \sim \mK_i} \ab[\ab(\sum_a \abs{P_{ja}(s) - P_{ia}(t)} \norm*{\vx_a}_2)^2]
    \\&\le \E_{i \sim \mu, j \sim \mK_i} \ab(\sum_a |P_{ja}(s)-P_{ia}(t)|)^2
    \\&\le 2\E_{i \sim \mu, j \sim \mK_i} \KL(P_j(s) \lVert P_i(t))
    \\&= 2(h(t)-h(s)).
\end{align}
The first equality uses stationarity and the variance identity; the second uses \eqref{eq:K transfer} and stationarity under $\mK$. The inequalities use the triangle inequality, $\norm{\vx_a}_2\le1$, and Pinsker's inequality for the unnormalized sum of absolute probability differences. The final equality is \eqref{eq:entropy}. For $s=1$, $V(1)=0$, so the preceding bound on $V(s)$ also holds at this endpoint.

Taking $t = t_\ell$ and $s = t_{\ell-1}$ therefore gives
\begin{align}
    V(t) \le \frac{2\log k}{r} + V(s) \le \frac{2\log k}{r} + \frac{2}{s}(\alpha + \eta).
\end{align}

\paragraph{Dividing by $\alpha$.} Starting from \eqref{eq:error}, taking $s=t$, rearranging, and dividing by $\alpha$ yields
\begin{align}
    \E_{j \sim P_i(t),a \sim P_j} \Delta_a =\frac{t}{\alpha(1-t)}(\vx_i - \bar{\vx}_i(t)) - \frac{1}{\alpha}\E_{j \sim P_i(t),a \sim P_j} \ve_j.
\end{align}
On the other hand, applying \eqref{eq:Ps-identity} to ($\Delta_1, \dots, \Delta_k$) gives
\begin{align}
    \E_{j \sim P_i(t)} \Delta_j - \E_{j \sim P_i(t), a \sim P_j} \Delta_a = \frac{t}{1-t} \ab(\Delta_i - \E_{j \sim P_i(t)} \Delta_j).
\end{align}
Adding and using the triangle inequality with $t\le1/2$, as ensured by our parameter choices below, and the bounds $\alpha\le1$ and $\norm{\vx}_2\le1$ for $\vx\in[0,1]^d$ gives
\begin{align}
    \norm*{\E_{j \sim P_i(t)} \Delta_j}_2 &\le \frac{2t}{\alpha}\norm{\vx_i - \bar{\vx}_i(t)}_2 + \frac{1}{\alpha}\E_{j \sim P_i(t),a \sim P_j}\norm{\ve_j}_2 + 2t \ab(\norm{\Delta_i}_2 + \E_{j \sim P_i(t)} \norm{\Delta_j}_2).
    \\&\le \frac{8t}{\alpha} + \frac{1}{\alpha}\E_{j \sim P_i(t),a \sim P_j}\norm{\ve_j}_2.
\end{align}
Averaging over $i \sim \mu$, we have
\begin{align}
    \E_{i \sim \mu} \norm*{\E_{j \sim P_i(t)} \Delta_j}_2 \le \frac{8t + \eta}{\alpha}\eqdef \gamma
\end{align}
This is the quantitative version of \eqref{eq:qualitative-efp}.

\paragraph{Using continuity of $f$.}
The final step is to use continuity of $f$ to conclude that we have found a fixed point. For independent draws $j,a\sim P_i(t)$, we have
\begin{align}
    \frac12 \E_{j, a \sim P_i(t)}\norm{\vx_j - \vx_a}_2^2 = \E_{j \sim P_i(t)} \norm{\vx_j - \bar{\vx}_i(t)}_2^2.
\end{align}
Here, $\bar{\vx}_i(t)$ is the mean of $\vx_j$ over $P_i(t)$. Therefore, we have
\begin{align}
    \E_{j \sim P_i(t)} \norm{\Delta_j}_2 &\le \sqrt{\E_{j \sim P_i(t)} \norm{\Delta_j}_2^2}
    \\&\le \norm*{\E_{j \sim P_i(t)}\Delta_j}_2 + \sqrt{\frac12 \E_{j, a \sim P_i(t)} \norm{\Delta_j - \Delta_a}_2^2}
    \\&\le \norm*{\E_{j\sim P_i(t)}\Delta_j}_2 + 2L\sqrt{\frac12 \E_{j, a \sim P_i(t)} \norm{\vx_j - \vx_a}_2^2}
    \\&\le \norm*{\E_{j\sim P_i(t)}\Delta_j}_2 + 2L\sqrt{\E_{j \sim P_i(t)} \norm{\vx_j - \bar{\vx}_i(t)}_2^2}.
\end{align}
where we use, in order: Cauchy-Schwarz; the inequality $\sqrt{a + b} \le \sqrt{a} + \sqrt{b}$ together with variance decomposition; Lipschitzness of $f$ with the bound $L \ge 1$; and finally the previous identity.
Taking expectations over $\mu$, using the fact that $\mu$ is stationary, and applying Jensen now yields
\begin{align}\label{eq:amplification residual}
\E_{i \sim \mu} \norm{\Delta_i}_2 &\le \gamma + 2L\sqrt{V(t)} \le \frac{8t + \eta}{\alpha} + 2L\sqrt{\frac{2\log k}{r} + \frac{2(\alpha + \eta)}{s}}.
\end{align}
We now take $\alpha_\ell=b^{-\ell}$ and set $t_\ell=\delta\alpha_\ell/64$, for a base $b\ge2$ to be chosen below. This makes $8t/\alpha=\delta/8$. Since $s=t_{\ell-1}$ and $t_0=1$, we also have $\alpha/s\le64/(\delta b)$ and $1/s\le64b^{r-1}/\delta$. Substituting in \eqref{eq:amplification residual} gives
\begin{align}
    \E_{i\sim\mu}\norm{\Delta_i}_2
    &\le\frac\delta8+\eta b^r
      +2L\sqrt{\frac{2\log k}{r}+\frac{128}{\delta b}+\frac{128\eta b^{r-1}}\delta}.
\end{align}
The remaining parameter choices bound the three terms under the square root. Taking $r=(1+\ceil{\log_2 K})\ceil{128(L+1)^2/\delta^2}$ bounds the entropy term by $\delta^2/(64(L+1)^2)$, since $k\le K$. Taking $b=\ceil{\max\{2,2^{14}(L+1)^2/\delta^3\}}$ bounds the next term by $\delta^2/(128(L+1)^2)$. Finally, we choose $\eps=\delta/(16rb^r)$. Recalling that $\eta=\sqrt r\,\eps$, this gives $\eta b^r\le\delta/16$ and bounds the last term under the square root by $\delta^2/(2048(L+1)^2)$. Thus,
\begin{align}
    \E_{i\sim\mu}\norm{\Delta_i}_2
    &\le\frac{3\delta}{16}
      +2L\sqrt{\frac{\delta^2}{64(L+1)^2}+\frac{\delta^2}{128(L+1)^2}+\frac{\delta^2}{2048(L+1)^2}}
      <\frac\delta2.
\end{align}
This shows that a supported point in this layer has residual less than $\delta/2$. For fixed $L,\delta$, we have $r=O(\log K)$ and $b = \Theta(1)$, so $\epsilon = \delta/(16 r b^r) \geq K^{-c_0}$ for a sufficiently large absolute constant $c_0$. The construction takes time polynomial in the input and $\log K$, and extracting a fixed point takes time polynomial in the returned support encoding. To obtain the formulation stated in the theorem, for a fixed $C \geq 1$, we take $K = d^{2 \lceil C \rceil}$.  The constructed map $g$ has dimension $D=dr=O(d\log d)$, so $D^C\le K$ for all sufficiently large $d$. Taking $c=2\ceil{C}c_0$, we also have $D^{-c}\le d^{-c}=K^{-c_0}\le\eps$. This proves~\Cref{th:fixed set hardness}.

\begin{remark}[Beyond explicit support]
    \label{rem:sampling representations}
    Instead of explicitly listing each point in the support, one may instead take the output of the algorithm to be a probabilistic Turing machine, which samples a stationary distribution $\mu$ corresponding to an $\eps$-fixed set. For a support bound $K\ge 2$ and accuracy $\eps=K^{-c}$ for a sufficiently large absolute constant $c$, the proof of \Cref{th:fixed set hardness}, applied with accuracy $\delta/2$, uses $O(\log K)$ layers and, with suitable parameter choices, gives $\E_{\vx\sim\mu}\norm{f(\vx^{(\ell)})-\vx^{(\ell)}}_2\le\delta/4$ for some layer $\ell$. By Markov's inequality, two independent samples contain a point with residual at most $\delta/2$ in that layer with probability at least $3/4$. Thus, we can test every layer to obtain a $\delta$-fixed point with probability at least $3/4$, without access to the transition matrix. As a result, the same proof establishes hardness under randomized reductions even for superpolynomial support, provided $\log K$ is polynomial in the input size.
\end{remark}

\begin{remark}\label{rem:fixed-set-fptas}
    The hardness in \Cref{th:fixed set hardness} applies when the accuracy is small enough relative to a prescribed support bound. Without that restriction, approximate fixed sets admit a simple FPTAS: start from any point $\vx_1=\mathbf{0}$, and at each step, set $\vx_{t+1} = f(\vx_t)$. With $O(1/\epsilon)$ iterations, the uniform distribution over that sequence of points is an $\epsilon$-fixed set.
\end{remark}

%% file: text/multilinear-games.tex
\section{Hardness of Correlated Equilibria in Multilinear Games}
\label{sec:polytope}

We next consider games in which each strategy set $\cX_i$ is a compact rational polytope, given by an explicit list of linear inequalities, and each utility is an explicit multilinear polynomial. Our goal is to prove~\Cref{th:multilinear hardness}, which we restate below.

\polytopemult*

The basic idea is to implement quadratic imitation utilities using linear costs on a larger strategy polytope, through the use of auxiliary coordinates. We then use a constant number of additional players to evaluate a polynomial fixed-point map through a multilinear game, using an additional gadget to enforce agreement among the copies.

\subsection{Exponential precision}

\begin{proof}[Proof of \Cref{th:multilinear hardness}(i)] To begin with, we convert the hard fixed-point problem into one with a polynomial function. This step is standard, but we include it for completeness.

\paragraph{A polynomial map.} It is known that satisfying a generalized circuit with fan-out two is \PPAD-hard at some constant tolerance $0<\tau<1/4$~\citep{Rubinstein15:Inapproximability}. Let $d$ be its number of nodes and $n \ge d$ its encoding length. Each gate will be replaced by a continuous function of its inputs. For arithmetic gates, we use the clipping operation to make sure that the output stays in $[0,1]$. For Boolean and comparison gates, we write $a,b\in[0,1]$ for the input values, and define
\begin{align}
    R_\tau(x) = \clip_{[0,1]}\ab(\frac{x-\tau}{1-2\tau}).
\end{align}
This rescaling ensures that approximately Boolean inputs yield exact Boolean outputs, leaving some slack for the subsequent approximation. We will use the following gates.\footnote{\citet{FilosRatsikas23:Complexity} have shown that $\PPAD$-hardness holds with only (truncated) addition and complement gates; that simplification is not important for our purposes here.}
\begin{itemize}
    \item Comparison: $h_{<}(a,b) = \clip_{[0,1]}\ab(\frac12+\frac{b-a}{2\tau})$.
    \item OR: $h_{\mathrm{OR}}(a,b) = R_\tau(\max\{a,b\})$.
    \item AND: $h_{\mathrm{AND}}(a,b) = R_\tau(\min\{a,b\})$.
    \item NOT: $h_{\mathrm{NOT}}(a) = 1-R_\tau(a)$.
\end{itemize}
For a constant $\lambda \ge1$, each gate map satisfies $|h_i(a,b)-h_i(a',b')|\le\lambda(|a-a'|+|b-b'|)$, and each variable occurs in at most three input positions. We will use a sufficiently accurate constant-degree polynomial approximation of each $h_i$, making sure that it preserves the range $[0,1]$ and has bounded derivatives. In particular, we make use of the Bernstein approximation~\citep{Lorentz12:Bernstein}:
\begin{align}
    f_i(a,b) = \E_{\xi_a \sim \operatorname{Bin}(p,a),\,\xi_b \sim \operatorname{Bin}(p,b)} h_i(\xi_a/p,\xi_b/p),
\end{align}
where the binomial variables above are independent. We will use the following basic properties of the Bernstein approximation~\citep{Lorentz12:Bernstein}; we include the short proofs for completeness.

\begin{claim}
    \label{cl:bernstein approximation}
    For every integer $p\ge1$, $f_i$ is a polynomial of degree at most $p$ in each input (for a total degree at most $2p$). For all $(a,b)\in[0,1]^2$,
    \begin{align}
        0\le f_i(a,b)\le1,\quad
        |f_i(a,b)-h_i(a,b)|\le\frac{\lambda}{\sqrt p},\quad
        |\partial_a f_i(a,b)|,\ |\partial_b f_i(a,b)|\le\lambda.
    \end{align}
\end{claim}

\begin{proof}
The range bound follows because $f_i$ averages values in $[0,1]$. The degree bound follows by expanding the binomial probabilities. Lipschitzness of $h_i$ together with the binomial variance bound give
\begin{align}
    |f_i(a,b)-h_i(a,b)|
    \le\lambda\ab(\E|\xi_a/p-a|+\E|\xi_b/p-b|)
    \le\lambda\ab(\sqrt{\frac{a(1-a)}p}+\sqrt{\frac{b(1-b)}p})
    \le\frac{\lambda}{\sqrt p}.
\end{align}
The derivative bounds follow by differentiation and the Lipschitznes of $h_i$.
\end{proof}

In what follows, we fix an integer $p\ge(4\lambda/\tau)^2$. By \Cref{cl:bernstein approximation}, the approximation error is at most $\tau/4$. Because every input appears at most three times, it follows that the Euclidean operator norm of $f$ is at most $\sqrt6\lambda\le3\lambda$, giving a Lipschitz bound $L=3\lambda$ for $f$ in the normalized norm as well.

We now set $\delta=\tau/(4d)$. If $\norm{f(\vx)-\vx}_2 \le \delta$, each coordinate residual is at most $\sqrt d\,\delta \le \tau/4$. Accounting for the approximation error, this satisfies every original gate at tolerance $\tau$. Thus, finding such a point is \PPAD-hard.

\paragraph{Amplifying the support bound.}
We next apply the amplification argument from \Cref{sec:hardness-fixedsets}. In particular, if $K \geq 2$ is the support, we set
\begin{align}\label{eq:multilinear amplification parameters}
    r \defeq (1+\ceil{\log_2 K})\ceil{128(L+1)^2/\delta^2} \text{ and } b \defeq \ceil{\max\{2,2\lambda+1,2^{14}(L+1)^2/\delta^3\}},
\end{align}
and define $g:[0,1]^{dr}\to[0,1]^{dr}$ by
\begin{align}
    g(\vx^{(1)},\dots,\vx^{(r)})
    = \ab((1-b^{-\ell})\vx^{(\ell)}+b^{-\ell}f(\vx^{(\ell)}))_{\ell=1}^r.
\end{align}
Moreover, $\zeta \defeq \delta/(16rb^r)$. The following bound follows directly from the argument in~\Cref{sec:hardness-fixedsets}.

\begin{claim}[Extracting a fixed point]\label{cl:multilinear amplification}
    If $\mu$ has support size at most $K$ and satisfies $\E[Q(\vx)-Q(g(\vx))]\le\zeta$ for every convex $1$-Lipschitz $Q$, then some $\vx\in\supp\mu$ and $\ell\in[r]$ satisfy $\norm{f(\vx^{(\ell)})-\vx^{(\ell)}}_2\le\delta$.
\end{claim}

\paragraph{Encoding the quadratic penalty into the constraints.} Let $T: [0, 1] \ni x \mapsto \min\{2x,2(1-x)\}$ and
\begin{align}
    p_m(x)=x-\sum_{j=1}^m4^{-j}T^{(j)}(x),
\end{align}
where $T^{(j)}$ denotes the $j$th composition of $T$ with itself. The function $p_m$ linearly interpolates $x^2$ on the mesh $2^{-m}\Z\cap[0,1]$, as illustrated in~\Cref{fig:quadratic-interpolation}. On a mesh interval $[a,b]$, its error is $(x-a)(b-x)$, so
\begin{align}\label{eq:multilinear square}
    x^2 \le p_m(x) \le x^2+4^{-(m+1)}.
\end{align}

\begin{figure}[htbp]
    \centering
    \input{tikz/quadratic-interpolation.tex}
    \caption{The piecewise-linear interpolants $p_1$ and $p_2$ relative to $x^2$. Each successive refinement doubles the number of pieces and divides the maximum error by four.}
    \label{fig:quadratic-interpolation}
\end{figure}
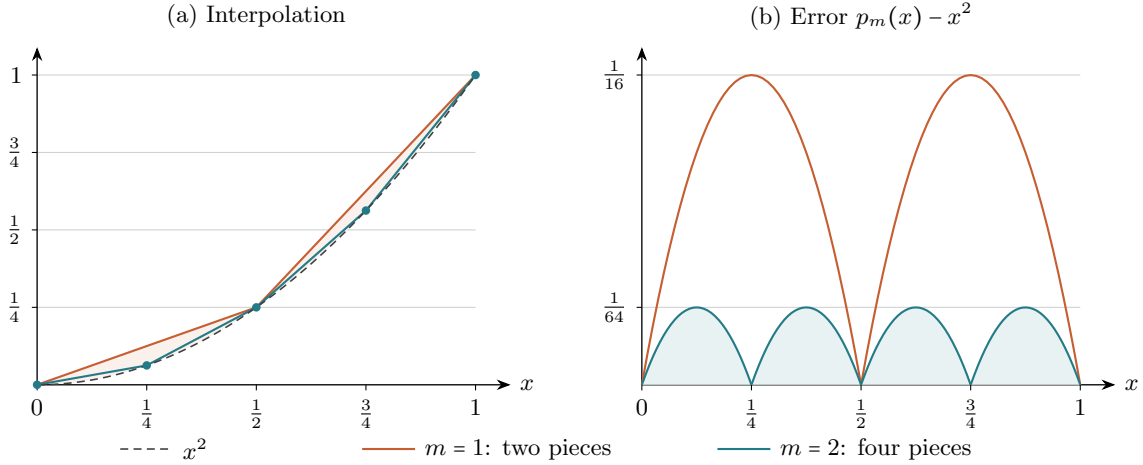

Although this interpolant has exponentially many pieces, it will suffice to consider the following compact polytope:
\begin{align}\label{eq:multilinear lift}
    \begin{gathered}
        0\le x,t,w_1,\dots,w_m\le1,\quad w_0=x,\\
        w_j\le2w_{j-1},\quad w_j\le2(1-w_{j-1})\quad \text{for } j\in[m],\quad
        t\ge x-\sum_{j=1}^m4^{-j}w_j.
    \end{gathered}
\end{align}
Crucially, in the game constructed below, $t$ enters the player's utility as a linear cost, so the player will be incentivized to choose a lift minimizing $t$. This choice recovers $p_m(x)$, as formalized below.

\begin{claim}
    \label{cl:multilinear lift}
    For every fixed $x\in[0,1]$, the minimum feasible $t$ in \eqref{eq:multilinear lift} is $p_m(x)$, attained when $w_j=T^{(j)}(x)$.
\end{claim}

\begin{proof}
We show that minimizing $t$ forces $w_j=T(w_{j-1})$. For the sake of contradiction, we consider the first coordinate $w_j$ for which $w_j < T(w_{j-1})$. If $w_j$ is increased by $\Delta$, later coordinates can be reduced to retain feasibility, losing at most $2^{s-j}\Delta$ at position $s > j$. Indeed, after setting $w_j'=w_j+\Delta$, recursively choosing $w_s'=\min\{w_s,T(w_{s-1}')\}$ preserves feasibility, and the $2$-Lipschitz property of $T$ implies $0 \le w_s-w_s'\le2^{s-j}\Delta$ by induction. As a result, the gain in $\sum_s4^{-s}w_s$ is at least
\begin{align}
    4^{-j}\Delta\ab(1-\sum_{h=1}^{m-j}2^{-h})>0,
\end{align}
which is a contradiction. We conclude that every minimizing lift satisfies $w_j=T^{(j)}(x)$ and $t=p_m(x)$, as claimed.
\end{proof}

Let $\cS$ be the product of $dr$ copies of \eqref{eq:multilinear lift}. We write a lifted strategy as $\vs=(\vx,\vt,\vw_1,\dots,\vw_m)\in\cS$, where each of these vectors lies in $[0,1]^{dr}$. We denote by $\vx(\vs)=\vx$ its projection onto the original coordinates and $c(\vs) \defeq (dr)^{-1}\sum_a t_a$ its linear cost. The basic property we need is that
\begin{align}\label{eq:multilinear cost}
    c(\vs)\ge\norm{\vx(\vs)}_2^2,
    \qq{and every $\vx$ has a lift with}
    c(\vs)\le\norm{\vx}_2^2+4^{-(m+1)}.
\end{align}
We have therefore replaced the quadratic cost by a linear one using only $O(drm)$ variables and inequalities.

\paragraph{Evaluating the map with copies.} To implement each $f_i(a,b)$ as a multilinear function, we write each binomial variable as a sum of $p$ independent Bernoulli variables controlled by different players. In particular, if coordinate $i$ has inputs $a_i,b_i$, we set
\begin{align}
    \widetilde f_i(\vz^{(1)},\dots,\vz^{(2p)})
    =\E h_i\ab(\frac1p\sum_{j=1}^p \xi_j,\frac1p\sum_{j=p+1}^{2p}\xi_j),
\end{align}
where the $\xi_j$ are independent, with success probabilities $z^{(j)}_{a_i}$ for $j\le p$ and $z^{(j)}_{b_i}$ otherwise. This multilinear polynomial takes values in $[0,1]$ even when the copies disagree, and equals $f_i(\vx)$ when all copies equal $\vx$.

Moreover, we use one more copy for the retained coordinate in each layer of $g$, for a total of $q=2p+1$ copies. Writing $\vz^{(j,\ell)}$ for the $\ell$th $d$-dimensional block of copy $j$, we define
\begin{align}
    \widetilde g^{(\ell)}(\vz^{(1)},\dots,\vz^{(q)})
    =(1-b^{-\ell})\vz^{(q,\ell)}
    +b^{-\ell}\widetilde f(\vz^{(1,\ell)},\dots,\vz^{(2p,\ell)}).
\end{align}
This map takes values in $[0,1]^{dr}$ and equals $g(\vx)$ when all copies equal $\vx$. Since each output uses at most $q$ inputs and each input occurs at most three times in a fixed copy, Cauchy-Schwarz yields
\begin{align}\label{eq:multilinear copies}
    \norm{\widetilde g(\vz^{(1)},\dots,\vz^{(q)})-g(\vx)}_2
    \le 3q\sqrt{\sum_{j=1}^q\norm{\vz^{(j)}-\vx}_2^2}.
\end{align}

\paragraph{The final construction.} The game contains $q+2$ players, each of whom has strategy set $\cS$. For a strategy profile $(\vs_x,\vs_y,\vs_1,\dots,\vs_q)$, we write $\vx=\vx(\vs_x)$, $\vy=\vx(\vs_y)$, and $\vz^{(j)}=\vx(\vs_j)$ for the projected vectors. All players use the cost function $c$ defined above.

The roles of $\vx$ and $\vy$ are inherited from the quadratic game in \Cref{sec:hardness-fixedsets}: $\vx$ tries to imitate $\vy$, while $\vy$ tries to imitate the map evaluated at $\vx$. To make the utilities multilinear, we first replace the quadratic costs by the lifted linear cost $c$ and then evaluate the map using $\widetilde g(\vz^{(1)},\dots,\vz^{(q)})$, where each copy player is imitating $\vx$. Moreover, we reward $\vx$ for matching the copies, assigning a weight $M\ge1$ per copy. When $M \gg$ 1 and the imitation is sufficiently accurate, each copy will be forced to be close to $\vx$. At the same time, we retain the incentive for $\vx$ to imitate $\vy$.

In this context, we define the utilities
\begin{align}\label{eq:multilinear utilities}
    u_x(\vs_x,\vs_y,\vs_1,\dots,\vs_q)
    &=\frac{2}{dr(1+qM)}\ip{\vx,\vy+M\sum_j\vz^{(j)}}-c(\vs_x),\\
    u_y(\vs_x,\vs_y,\vs_1,\dots,\vs_q)
    &=\frac{2}{dr}\ip{\vy,\widetilde g(\vz^{(1)},\dots,\vz^{(q)})}-c(\vs_y),\\
    u_j(\vs_x,\vs_y,\vs_1,\dots,\vs_q)
    &=\frac{2}{dr}\ip{\vz^{(j)},\vx}-c(\vs_j)\quad \text{for } j\in[q].
\end{align}
(The factors $(dr)^{-1}$ are used per our normalized norm convention.) Each utility is multilinear and lies in $[-1,2]$. Shifting and scaling would place the utilities in $[0,1]$, changing the CE accuracy by only a constant factor, so we can work directly with the utilities above.

\begin{claim}
    \label{cl:multilinear extraction}
    Let $\mu$ be an $\eps$-correlated equilibrium of the constructed game and $\rho=\sqrt{\eps+4^{-(m+1)}}$. The $\vx$-marginal of $\mu$ is an $\eta$-convex EFP of $g$, where
    \begin{align}\label{eq:multilinear extraction}
        \eta=2(1+qM)\rho+3q\sqrt{M^{-1}+(2q+1)\rho}.
    \end{align}
\end{claim}

\begin{proof}
We define the deviations
\begin{align}
    \widehat{\vx}=\E\ab[\frac{\vy+M\sum_j\vz^{(j)}}{1+qM}\,\middle|\,\vs_x], \quad
    \widehat{\vz}^{(j)}=\E[\vx\mid\vs_j], \quad
    \widehat{\vy}=\E[\widetilde g(\vz^{(1)},\dots,\vz^{(q)})\mid\vs_y].
\end{align}
For each recommendation $\vs_x$, we let $\vs'_x$ be the minimum-cost lift of $\widehat{\vx}$. \Cref{cl:multilinear lift} gives
\begin{align}
    \vx(\vs'_x)=\widehat{\vx},\quad c(\vs'_x)\le\norm{\widehat{\vx}}_2^2+4^{-(m+1)}.
\end{align}
Writing $\vs_{-x}$ for the other players' recommendations, the conditional utility gain is
\begin{align}
    \E\ab[u_x(\vs'_x,\vs_{-x})-u_x(\vs_x,\vs_{-x})\,\middle|\,\vs_x] &=\frac{2}{dr}\ip{\widehat{\vx}-\vx,\widehat{\vx}}+c(\vs_x)-c(\vs'_x)\\
    &\ge 2\norm{\widehat{\vx}}_2^2-\frac{2}{dr}\ip{\vx,\widehat{\vx}}+c(\vs_x)-\norm{\widehat{\vx}}_2^2-4^{-(m+1)}\\
    &=\norm{\vx-\widehat{\vx}}_2^2+c(\vs_x)-\norm{\vx}_2^2-4^{-(m+1)}\\
    &\ge\norm{\vx-\widehat{\vx}}_2^2-4^{-(m+1)},
\end{align}
where the first equality uses the definition of $\widehat{\vx}$ and the last inequality uses $c(\vs_x)\ge\norm{\vx}_2^2$. Taking expectations and using the fact that $\mu$ is an $\epsilon$-correlated equilibrium,
\begin{align}
    \E\norm{\vx-\widehat{\vx}}_2^2 \leq 4^{-(m+1)}  +\eps = \rho^2.
\end{align}
Similarly, we obtain
\begin{align}\label{eq:multilinear obedience}
    \E\norm{\vx-\widehat{\vx}}_2^2\le\rho^2,\quad
    \E\norm{\vz^{(j)}-\widehat{\vz}^{(j)}}_2^2\le\rho^2,\quad
    \E\norm{\vy-\widehat{\vy}}_2^2\le\rho^2.
\end{align}
We now argue that the copies are close to $\vx$ in mean square. Indeed, we have
\begin{align}
    (1+qM)\E\ip{\vx,\widehat{\vx}} =\E\ip{\vx,\vy}+M\sum_j\E\ip{\vx,\vz^{(j)}}, \quad \E\ip{\vz^{(j)},\widehat{\vz}^{(j)}} = \E\ip{\vz^{(j)},\vx}.
\end{align}
Thus,
\begin{align}
    M\sum_j\E\norm{\vz^{(j)}-\vx}_2^2
    &= \frac1{dr}\E\ip{\vx,\vy-\vx}
    +\frac{1+qM}{dr}\E\ip{\vx,\vx-\widehat{\vx}} + \frac{M}{dr}\sum_j\E\ip{\vz^{(j)},\vz^{(j)}-\widehat{\vz}^{(j)}} \\
    &\le 1+(1+2qM)\rho.
\end{align}
To establish the inequality, we bound each of the three terms separately. First, since $\vx,\vy\in[0,1]^{dr}$, we have $\norm{\vx}_2\le1$ and $\norm{\vy-\vx}_2\le1$ in the normalized norm, so
\begin{align}
    \frac1{dr}\E\ip{\vx,\vy-\vx}
    \le\E\ab[\norm{\vx}_2\norm{\vy-\vx}_2]\le1.
\end{align}
Next, for the other two terms, Cauchy-Schwarz together with \eqref{eq:multilinear obedience} imply
\begin{align}
    \frac1{dr}\ab|\E\ip{\vx,\vx-\widehat{\vx}}|
    &\le\sqrt{\E\norm{\vx}_2^2}\sqrt{\E\norm{\vx-\widehat{\vx}}_2^2}\le\rho,\\
    \frac1{dr}\ab|\E\ip{\vz^{(j)},\vz^{(j)}-\widehat{\vz}^{(j)}}|
    &\le\sqrt{\E\norm{\vz^{(j)}}_2^2}\sqrt{\E\norm{\vz^{(j)}-\widehat{\vz}^{(j)}}_2^2}\le\rho.
\end{align}
It remains to establish the convex EFP condition. For a convex $1$-Lipschitz function $Q$, \eqref{eq:multilinear obedience} implies
\begin{align}
    \E Q(\vx)
    &\le\E Q(\widehat{\vx})+\E\norm{\vx-\widehat{\vx}}_2
    \le\E Q(\widehat{\vx})+\sqrt{\E\norm{\vx-\widehat{\vx}}_2^2}
    \le\E Q(\widehat{\vx})+\rho.
\end{align}
Next, applying Jensen's inequality,
\begin{align}
    Q(\widehat{\vx}) \le\E\ab[Q\ab(\frac{\vy+M\sum_j\vz^{(j)}}{1+qM})\,\middle|\,\vs_x] \le\E\ab[\frac{Q(\vy)+M\sum_j Q(\vz^{(j)})}{1+qM}\,\middle|\,\vs_x].
\end{align}
Taking expectations,
\begin{align}
    \label{eq:EQ-ineq}
    \E Q(\vx)\le\frac{\E Q(\vy)+M\sum_j\E Q(\vz^{(j)})}{1+qM}+\rho.
\end{align}
Similarly, we obtain
\begin{align}
    \E Q(\vz^{(j)})
    &\le\E Q(\widehat{\vz}^{(j)})+\rho
    = \E Q(\E[\vx\mid\vs_j])+\rho
    \le\E Q(\vx)+\rho,\\
    \E Q(\vy)
    &\le\E Q(\widehat{\vy})+\rho
    \le\E Q(\widetilde g(\vz^{(1)},\dots,\vz^{(q)}))+\rho.
\end{align}
Combining the first inequality with~\eqref{eq:EQ-ineq},
\begin{align}
    (1+qM) \E Q(\vx)
    \le \E Q(\vy) + qM\E Q(\vx)+ (1+2qM)\rho \iff \E Q(\vx) \leq \E Q(\vy) + (1+2qM)\rho.
\end{align}
Using the bound for $\E Q(\vy)$,
\begin{align}
    \E Q(\vx)\le\E Q(\widetilde g(\vz^{(1)},\dots,\vz^{(q)}))+2(1+qM)\rho.
\end{align}
Finally, Lipschitzness of $Q$, \eqref{eq:multilinear copies}, and Cauchy-Schwarz yield
\begin{align}
    \E\ab[Q(\widetilde g(\vz^{(1)},\dots,\vz^{(q)}))-Q(g(\vx))] \le3q\sqrt{\sum_j\E\norm{\vz^{(j)}-\vx}_2^2} \le3q\sqrt{M^{-1}+(2q+1)\rho}.
\end{align}
Combining these bounds gives $\E[Q(\vx)-Q(g(\vx))]\le\eta$.
\end{proof}
In other words, increasing $M$ enforces agreement between the copies, while decreasing $\eps$ makes imitation sufficiently accurate.

\paragraph{Choosing the parameters.} To conclude, we define the parameters as
\begin{align}\label{eq:multilinear precision}
    M=\ceil{64(3q+1)^2/\zeta^2},\quad
    \rho_0=\frac{\zeta}{32(q+1)M},\quad
    \eps=2^{-\ceil{\log_2(4/\rho_0^2)}},
\end{align}
and choose the least positive $m$ such that $4^{-(m+1)}\le\rho_0^2/4$. The choice of $\eps$ guarantees that $\rho\le\rho_0$. The first term in \eqref{eq:multilinear extraction} satisfies
\begin{align}
    2(1+qM)\rho
    \le\frac{\zeta}{16}\frac{1+qM}{(q+1)M}
    \le\frac{\zeta}{16},
\end{align}
where the last inequality uses $M\ge1$. For the second term, since $\zeta\le1$, we have
\begin{align}
    (2q+1)\rho
    \le\frac{(2q+1)\zeta}{32(q+1)M}
    \le\frac1M,
\end{align}
so
\begin{align}
    3q\sqrt{M^{-1}+(2q+1)\rho}
    \le3q\sqrt{\frac2M}
    \le\frac{\sqrt2\,\zeta}{8}\frac{3q}{3q+1}
    \le\frac{\sqrt2\,\zeta}{8}.
\end{align}
Thus, $\eta < \zeta$. By \Cref{cl:multilinear extraction,cl:multilinear amplification}, any such correlated equilibrium supported on at most $2^n$ points yields a $\delta$-fixed point of $f$. Since $\delta^{-1}=O(d)$, $r$ and $b$ are polynomial in $n$, and $\log(1/\zeta)=O(r\log b+\log(r/\delta))=\operatorname{poly}(n)$. As a result, $m$, $\log M$, and $\log(1/\eps)$ are polynomial in $n$. The polytope has $O(drm)$ variables and inequalities, and the utilities contain $O(dr)$ monomials. Thus, $N=\operatorname{poly}(n)$. For every fixed $C$, we have $N^C\le2^n$ for all sufficiently large $n$.
\end{proof}

\subsection{Polynomial precision under PCP-for-PPAD}

\begin{proof}[Proof of \Cref{th:multilinear hardness}(ii)]
We now operate under \Cref{ass:pcp-ppad}. We follow the same construction, with a minor change in the input circuit and the parameters.

We start from the circuit corresponding to \Cref{ass:pcp-ppad}, assumed to have $d$ nodes and encoding length $n \ge d$. We use again the Bernstein approximation with $p\ge(4\lambda/\tau)^2$, and choose a parameter $\delta$ such that $\delta^2\le\gamma\tau^2/64$. If $\norm{f(\vx)-\vx}_2\le\delta$, then
\begin{align}
    \frac1d\big|\{i:|f_i(\vx)-x_i|>\tau/4\}\big|
    \le\frac{16\delta^2}{\tau^2}\le\frac\gamma4.
\end{align}
Every other coordinate satisfies $|h_i(\vx)-x_i|\le\tau/2$ by the Bernstein error bound, so the corresponding gate is satisfied at tolerance $\tau$. The assignment therefore satisfies at least $(1-\gamma)d$ gates and can be decoded in polynomial time using \Cref{ass:pcp-ppad}. This shows that, under~\Cref{ass:pcp-ppad}, the problem is \PPAD-hard even at constant accuracy $\delta = \Theta(1)$.

The rest of the steps are the same. To begin with, we choose $r$ and $b$ per~\eqref{eq:multilinear amplification parameters}, and $\zeta=\delta/(16rb^r)$. Now $b$ is constant, $r=O(\log n)$, and $\zeta=n^{-O(1)}$. Moreover, the parameter choices in~\eqref{eq:multilinear precision} give $M=\Theta(\zeta^{-2})$, $\rho_0=\Theta(\zeta^3)$, $\eps=\Theta(\zeta^6)=n^{-O(1)}$, $m=O(\log n)$. As a result, there are $O(d \log^2 n)$ variables and inequalities and $O(d\log n)$ monomials in each utility.
\end{proof}

%% file: tikz/quadratic-interpolation.tex
\begingroup
\definecolor{interpOrange}{HTML}{C65E32}
\definecolor{interpTeal}{HTML}{237B87}
\begin{tikzpicture}[
    font=\small,
    >=Stealth,
    axis/.style={->,black,line width=0.5pt},
    guide/.style={black!18,line width=0.4pt},
    quadratic/.style={black!75,densely dashed,line width=0.6pt},
    coarse/.style={interpOrange,line width=0.8pt},
    fine/.style={interpTeal,line width=0.8pt}
]
    \begin{scope}[x=5.8cm,y=4.1cm]
        \node at (0.5,1.19) {(a) Interpolation};
        \foreach \v in {0.25,0.5,0.75,1} {
            \draw[guide] (0,\v) -- (1,\v);
        }
        \fill[interpOrange!9]
            (0,0) -- (0.5,0.25) -- (1,1)
            -- plot[domain=1:0,samples=81,variable=\t] (\t,{\t*\t})
            -- cycle;
        \draw[axis] (0,0) -- (1.08,0) node[right] {$x$};
        \draw[axis] (0,0) -- (0,1.09);
        \foreach \v/\lab in {0/0,0.25/{\tfrac14},0.5/{\tfrac12},0.75/{\tfrac34},1/1} {
            \draw[black] (\v,0) -- ++(0,-0.018)
                node[below] {$\lab$};
        }
        \foreach \v/\lab in {0.25/{\tfrac14},0.5/{\tfrac12},0.75/{\tfrac34},1/1} {
            \draw[black] (0,\v) -- ++(-0.012,0)
                node[left] {$\lab$};
        }
        \draw[coarse] (0,0) -- (0.5,0.25) -- (1,1);
        \draw[fine] (0,0) -- (0.25,0.0625) -- (0.5,0.25)
            -- (0.75,0.5625) -- (1,1);
        \draw[quadratic] plot[domain=0:1,samples=81,variable=\t]
            (\t,{\t*\t});
        \foreach \v in {0,0.25,0.5,0.75,1} {
            \fill[interpTeal] (\v,{\v*\v}) circle[radius=1.6pt];
        }
    \end{scope}

    \begin{scope}[xshift=8cm,x=5.8cm,y=65.6cm]
        \node at (0.5,0.074375) {(b) Error $p_m(x)-x^2$};
        \foreach \v/\lab in {0.015625/{\tfrac1{64}},0.0625/{\tfrac1{16}}} {
            \draw[guide] (0,\v) -- (1,\v);
            \draw[black] (0,\v) -- ++(-0.012,0)
                node[left] {$\lab$};
        }
        \draw[axis] (0,0) -- (1.08,0) node[right] {$x$};
        \draw[axis] (0,0) -- (0,0.068125);
        \foreach \v/\lab in {0/0,0.25/{\tfrac14},0.5/{\tfrac12},0.75/{\tfrac34},1/1} {
            \draw[black] (\v,0) -- ++(0,-0.001125)
                node[below] {$\lab$};
        }
        \foreach \j in {0,1} {
            \pgfmathsetmacro{\a}{\j/2}
            \pgfmathsetmacro{\b}{(\j+1)/2}
            \draw[coarse] plot[domain=\a:\b,samples=41,variable=\t]
                (\t,{(\t-\a)*(\b-\t)});
        }
        \foreach \j in {0,1,2,3} {
            \pgfmathsetmacro{\a}{\j/4}
            \pgfmathsetmacro{\b}{(\j+1)/4}
            \fill[interpTeal!10] (\a,0)
                -- plot[domain=\a:\b,samples=31,variable=\t]
                    (\t,{(\t-\a)*(\b-\t)})
                -- (\b,0) -- cycle;
            \draw[fine] plot[domain=\a:\b,samples=31,variable=\t]
                (\t,{(\t-\a)*(\b-\t)});
        }
    \end{scope}

    \draw[quadratic] (1.1,-0.85) -- (1.8,-0.85)
        node[right,black] {$x^2$};
    \draw[coarse] (4.3,-0.85) -- (5,-0.85)
        node[right,black] {$m=1$: two pieces};
    \draw[fine] (9,-0.85) -- (9.7,-0.85)
        node[right,black] {$m=2$: four pieces};
\end{tikzpicture}
\endgroup

%% file: text/ai-disclosure.tex
\section*{AI Disclosure}

The proofs of \Cref{th:fixed set hardness} and~\Cref{th:multilinear hardness} were AI-assisted. We have been examining the complexity of fixed sets for the best part of the last year. We were able to show that constant-support fixed sets are \PPAD-hard in low dimensions, and have been trying to strengthen this in higher dimensions through some gap amplification argument. With this as input, Astra GPT-6.0 provided the multi-scale argument that led to~\Cref{th:fixed set hardness}. We then significantly simplified and rewrote the proof, isolating \Cref{prop:non-quantitative} as an intuitive roadmap. For~\Cref{th:multilinear hardness}, Astra GPT-6.0 suggested the use of the piecewise linear interpolator, and helped us formalize some of the intermediate steps. All writing is due to the authors, and we take full responsibility for the correctness of the results.